\documentclass[a4paper,11pt]{article}
\usepackage[latin1]{inputenc}
\usepackage{amsfonts}
\usepackage{amsmath}
\usepackage{amssymb}
\usepackage{amsthm}
\usepackage[T1]{fontenc}
\usepackage[dvips]{graphicx}
\usepackage{color}
\usepackage{epsfig}
\usepackage{geometry}
\theoremstyle{plain}
\newtheorem{thm}{Theorem}
\newtheorem{lem}{Lemma}
\newtheorem{Def}{Definition}

\newtheorem{pro}{Proposition}
\newtheorem{rem}{Remark}
\newtheorem*{def*}{Definition}
\newtheorem*{prob*}{Direct Monodromy Problem}
\newtheorem*{lem*}{\textsc{Lemma}}
\newtheorem*{cor*}{\textsc{Corollary}}
\newtheorem*{con*}{\textsc{Conjecture}}
\newcommand{\e}{\varepsilon}
\newcommand{\bb}[1]{\mathbb{ #1 }}

\newcommand{\du}[1]{\frac{\delta #1}{\delta u}}

\newcommand{\eu}{{\cal{E}}}
 
\newcommand{\beq}{\begin{equation}}
\newcommand{\eeq}{\end{equation}}
\newcommand{\de}{\partial}

\newtheorem{exa}{Example}
\newcommand{\norm}[1]{\left| #1 \right|}

\title{A deformation of the method of characteristics and the Cauchy problem for Hamiltonian PDEs in the small dispersion limit}
\author{D.~Masoero\thanks{Grupo de F\'isica Matem\'atica,
Universidade de Lisboa, Av. Prof. Gama Pinto, 2 1649-003 Lisboa,  Portugal.  dmasoero@gmail.com }\,\,\, and A.~Raimondo\thanks{SISSA, Via Bonomea 265, 34136 Trieste, Italy.  andrea.raimondo@sissa.it}}
\date{}
\begin{document}

\maketitle

\abstract{We introduce a deformation of the method of characteristics valid for Hamiltonian perturbations
of a scalar conservation law in the small dispersion limit. Our method of analysis is based on the
\lq variational string equation', a functional-differential relation originally introduced by Dubrovin in a particular case,
of which we lay the mathematical foundation.

Starting from first principles, we construct the string equation explicitly up to
the fourth order in perturbation theory, and we show that the solution to the Cauchy problem of the Hamiltonian PDE satisfies
the appropriate string equation in the small dispersion limit. We apply our construction to  explicitly compute the first
two perturbative corrections of the solution to the general Hamiltonian PDE. In the KdV case, we prove the existence of
a quasi-triviality transformation at any order and for arbitrary initial data.}

\section*{Introduction}
The present paper is devoted to the study of Hamiltonian perturbations of the scalar quasilinear conservation law
\beq\label{unpert}
u_{t}=a(u)\,u_{x},
\eeq
where $u=u(x,t),\,x,t\in\bb{R}$, and $a$ is a non-constant function of $u$. This equation admits a formal Hamiltonian representation, given by
$$u_{t}(x)=\left\{u(x),H^{0}_{h}\right\}, \qquad \quad H^{0}_{h}:=\int h(u) dx, \quad h''=a,$$
where the Poisson bracket is the Gardner-Zakharov-Faddeev bracket: 
\beq\label{deltapb}
\left\{H,K\right\}:=\int \du{H}\frac{d}{dx}\du{K}\,dx,
\eeq
for every pair of functionals $H, K$. The variational derivative $\du{}$  is defined for any local functional $H=\int h(u,u_{x},u_{xx},\dots) dx$ through the Euler-Lagrange
operator $\mathcal{E}$, namely:
\begin{equation}\label{eulerlagrange}
\mathcal{E}(h):=\sum_{k\geq 0}(-1)^{k}\frac{d^{k}}{d x^{k}}\frac{\de h}{\de u^{(k)}}=:\du{H} .
\end{equation}
The standard method of characteristics provides local classical solutions of the Cauchy problem for \eqref{unpert} by solving the implicit equation
\beq\label{intro:hodo}
x+t\,a(u)=g'(u),
\eeq
where $g'$ is a local inverse of the initial data. The time lifespan of the solutions for generic initial data is finite: there exists a $t_{c}<\infty$, known as the \emph{critical time}, beyond which no classical solution exists. The point $(x_{c},t_{c})$ where the solution breaks is known as the \emph{critical point} of the solution.\\

The class of Hamiltonian perturbations of equation \eqref{unpert} considered in the present work is obtained by keeping fixed the Poisson bracket \eqref{deltapb},
and formally deforming the Hamiltonian $H_{h}^{0}$ in the following way:
$$H_{h}:=H_{h}^{0}+\e^{2}H^{2}+\e^{4}\,H^{4}+\dots \,.$$
Here, $\e$ is a (small) parameter, every functional $H^{k}$ is required to be local, and
the corresponding density to be a homogeneous polynomial in the derivatives $u^{(j)}, \,j=1,2,\dots$, of total degree $k$, where $\text{deg}\,u^{{j}}=j$.
As an example, the KdV equation $u_t=uu_x+\e^2 u_{xxx}$ fits in this class by choosing $h=\frac{u^{3}}{6}$, $H^{2}=-\frac{1}{2}\int u_{x}^{2}dx$, and $H^{k}=0,\,k>2.$\\

The above class of Hamiltonian perturbations has been studied in detail by Dubrovin in \cite{du06}, where in particular he provides a complete
classification (up to order $\e^{4}$) of families of commuting  Hamiltonians, and he formulates an important universality conjecture on the
critical behaviour of Hamiltonian perturbations. According to the conjecture, for any perturbed PDE and for generic initial data the solution
of the Cauchy problem exists up to some time, bigger than the critical time of the unperturbed equation. Moreover, around the critical point,
every such solution is locally described by a particular solution of a $4^{th}$-order ODE of Painlev\'e type. The universality conjecture was proved
to hold for any equation of the KdV hierarchy,  for particular classes of solutions \cite{claeys2009, clgr11}. One of the key steps for the
derivation of the universality conjecture obtained in \cite{du06} is the introduction, for a special subclass of perturbations, of a formal
identity, which in the KdV case takes the form 
\begin{align}\nonumber
x+u\,t=&g'+\e^{2}\left(\frac{1}{2}g^{(4)}u_{x}^{2}+g'''u_{xx}\right)+\e^{4}\left(\frac{3}{5}g^{(6)}u_{x}^{2}u_{xx}\right.\\ \label{int:kdvstring}
&\left.+\frac{9}{10}g^{(5)}u_{xx}^{2}+\frac{6}{5}g^{(5)}u_{x}u_{xxx}+\frac{3}{5}g^{(4)}u_{xxxx}-\frac{1}{24}g^{(7)}u_{x}^{4}\right)+O(\e^{6}) \; .
\end{align}
This equation clearly reminds of the characteristics equation \eqref{intro:hodo}: the aim of the present paper is to construct an analogue of
the above formal identity for a more general class of Hamiltonian perturbations, and to rigorously prove its range of validity. It turns
out that both problems are non-trivial. As an application, we use the new equation to explicitly find the first two corrections of the
semiclassical expansion of the solution to the perturbed PDE. It should be noted that in the case $g=u^{n}$, $n\in\mathbb{N}$,
then identity \eqref{int:kdvstring} reduces, modulo higher order terms, to the \emph{string equation} of the KdV hierarchy \cite{no90,mo90}. By analogy, and considering
that our procedure to obtain the equation follows a variational formulation, the new equation to be constructed will be  named
\lq \emph{variational string equation}\rq.\\

The paper is divided into two main parts, one more algebraic in nature, and a second one  containing analytic results.  After reviewing some
known facts on Hamiltonian perturbations in Section \ref{sec:dztheory}, we consider in Section \ref{sec:string} the construction of the variational
string equation. This is done according to the following scheme: first, we require the string equation to describe the stationary points of a certain
(approximate) conserved quantity, namely to be of the form
$$\du{S_{f}}=0,$$ 
for some suitable functional $S_{f}$. Besides this, at $\e=0$ we want to recover the characteristics equation \eqref{intro:hodo}, and at $t=0$ the
equation has to provide the desired initial data. While the latter two conditions are very natural, the motivation for the first one comes from
the following result, obtained by Lax in \cite{lax75}\footnote{An analogous result was proved independently by Novikov \cite{nov74},
by considering symmetries \\ rather than conserved quantities. In the Hamiltonian case the two approaches are equivalent.}: given a (nonlinear) evolutionary equation, together with a conserved quantity, then the space of stationary points of the conserved quantity is invariant under the flow. As shown in Section \ref{sec:dztheory}, equation \eqref{intro:hodo} fits in this picture; we require the variational string equation to be of the same nature.

As a byproduct of the above construction, we obtain a characterization of the Hamiltonian perturbations admitting a variational string equation: at order $O(\e^{2})$ the string equation exists for every perturbation, while at order $O(\e^{4})$ we find that a certain functional constraint has to be satisfied. On the other hand, if for a given perturbation a string equation exists (at some order), then it is unique.\\

In Section \ref{sec:solutions} we describe how to use the variational string equation in order to provide approximate solutions - for small times - to the perturbed equations. Given the solution $u$ of a Cauchy problem for an Hamiltonian perturbation, with $\e$-independent initial data, then at a formal level the coefficients $v^{i}, i\geq 1$, of the semiclassical expansion 
$$u=v^{0}+\e^{2}\,v^{1}+\e^{4}\,v^{2}+\dots$$
can be written in terms of the solution of the unperturbed equation $v^{0}$ (with same initial data). If the  Hamiltonian perturbation admits a string equation at order $2k$, then the first $k$ coefficients can be  explicitly found by elementary manipulations. This procedure was already used in our previous paper \cite{mara11}, where we wrote the result for the first coefficient $v^{1}$ without explaining the derivation.  In the present work we show the details of the proof,  and we compute the second term $v^{2}$.\\

Finally, Section \ref{sec:validity} is devoted to the rigorous proof of the validity of the variational string equation. We show that, given the solution $u(x,t,\e)$ of a Cauchy problem for a Hamiltonian perturbation which admits a variational string equation (at some order in $\e$), then there exists an integer $K$ and a domain in the $(x,t)$-plane where the variational string equation is satisfied up to order $\e^{K}$, that is:
$$\du{S_{f}}_{|u=u(x,t,\e)}=O(\e^{K}),$$
for $x$ and $t$ belonging to the domain. The number $K$ depends both on the regularity of the solution of the Cauchy problem and on the order in $\e$ of the string equation. The problem is non standard, for the functional $S_{f}$ is in general ill-defined and the variational string equation is defined only locally. We approach it by first showing that the variational derivative of a formal functional satisfies a \emph{linear} PDE with forcing term (this result, already interesting by itself, is a generalization of a result by Lax), and then describing in detail the domain of the variational string equation. After that, we are able to control the $L^{2}$-norm of $\du{S_{f}}$, uniformly in $x$ and $t$, and to show that it is small in $\e$ as required. Since the value of $\du{S_{f}}$ at the boundary of the domain is not known, an important ingredient of the proof is the introduction of a family of curves in the $(x,t)$-plane, suitable deformations of the characteristic lines.\\

The results of Section \ref{sec:validity} are valid provided the solution of the Cauchy problem is regular enough in the semiclassical limit.
Therefore, using the results we proved in  \cite{mara11}, they can be applied to Hamiltonian perturbations of generalized KdV type.
In Section \ref{sec:examples} we briefly review these results, and we prove that all the perturbative corrections of the Cauchy problem of KdV are
rational functions of the solution of the unperturbed equation and its derivatives.
This property, called quasi-triviality, was proved so far just for solutions of the unperturbed equation with never vanishing first derivative and for initial data with a
depending on $\e$ in a prescribed way \cite{du06,lizh06}.

\section{Hamiltonian perturbation of quasilinear PDEs}\label{sec:dztheory}
Here we briefly review the Dubrovin--Zhang construction \cite{duzh01} of Hamiltonian perturbations for quasilinear conservation laws.
We focus on the scalar case, which is treated in detail in \cite{du06}. In order to fix the notation, we recall first some well known
fact about the unperturbed case; we then consider the theory of Hamiltonian perturbations, and finally we discuss the problem of finding
approximate solutions (for small times) to the perturbed equations.

\subsection{Quasilinear conservation laws}
Consider the Cauchy problem for a scalar conservation law
\begin{equation}\label{eq:hopfgen}
v_{t}=a(v)\,v_{x},\qquad v(x,0)=\varphi(x)\, ,
\end{equation}
with $a(v)$ any non-constant smooth function. 
We look for classical solutions of (\ref{eq:hopfgen}), that is for a bounded differentiable function,
with bounded space derivative, such that its partial derivatives satisfy (\ref{eq:hopfgen}) identically
on $\bb{R}\times[0,t^{\ast}[$ for some positive time $t^{\ast}$. The standard way to construct solutions of the Cauchy problem is via the method
of characteristics. The characteristic lines $x(t)$ are defined by the equation 
\begin{equation}\label{eq:chargen}
\dot{x}(t)=- a(v(x(t),t))  \; ,
\end{equation}
and any classical solution is constant along these lines, which are mutually distinct provided $t^{\ast}\leq t_c$, where
 $$t_c=\inf_{\left\lbrace x | \varphi_x a'(\varphi)>0\right\rbrace}\frac{1}{\varphi_xa'(\varphi)} >0  \; .$$
The time $t_{c}$ is known as the \emph{critical time} and at that time, the solution is said to have a \emph{gradient catastrophe}: the solution
remains bounded while its derivatives blow up. Being constant along the characteristic lines, the solution of the Cauchy problem \eqref{eq:hopfgen}
can locally be described by the functional equation
\begin{equation}\label{eq:func}
x+t\,a(v)=g'(v)  \, ,
\end{equation}
where $g$ is a function obtained by the implicit equation $g'(\varphi(x))=x$.  Equation (\ref{eq:func})
admits a variational formulation, constructed in the following way:  it is well known that for any function $f$, the (formal)  functional $H^{0}_{f}=\int f(u)dx$ is conserved along the flow of
\eqref{eq:hopfgen}. In addition, one has that the family of functionals
$$Q_{\alpha,\beta,f}:=\int \left(x\,\alpha(v)+t\,\beta(v)-f(v)\right) dx \,,$$
is conserved, for any choice of $f$, provided the relation
$$\beta'(v)=\alpha'(v)\,a(v),$$
is satisfied. Choosing $f$ such that $f'=\alpha' \, g'$, then \eqref{eq:func} is equivalent to 
$$\frac{\delta Q_{\alpha,\beta,f}}{\delta v}=0,$$
provided $\alpha'$ does not vanish.  Note that in this unperturbed case one can always choose $\alpha(v)=v$, thus obtaining \eqref{eq:func} directly.

\subsection{Hamiltonian perturbations}
Let $h(u)$, $c(u)$, $p(u)$ and $s(u)$ be smooth functions of one variable, with $c$ not identically zero. Then, introduce the quantities
\begin{align*}
c_{h}&:=c\,h''',\\
p_{h}&:=p\,h'''+\frac{3}{10}\,c^{2}\,h^{(4)},\\
s_{h}&:=s\,h'''-\frac{c\,c''}{8}\,h^{(4)}-\frac{c\,c'}{8}\,h^{(5)}-\frac{c^{2}}{24}h^{(6)}-\frac{p'}{6}h^{(4)}-\frac{p}{6}h^{(5)},
\end{align*}
and consider the linear map
defined by
\begin{gather}
D_{(c,p,s)}(h):=h-\frac{\e^2}{2}\,c_{h}(u)\,u_{x}^2+
\e^4\Big(p_{h}(u)\, u_{xx}^2+s_{h}(u)\,u_{x}^4 \Big),\label{conden}
\end{gather}
where $\e$ is a real parameter. The map \eqref{conden} has been introduced by Dubrovin in \cite{dub08}, where he called it \lq $D-$operator\rq.
Note that when applied to a quadratic function, \eqref{conden} reduces to the identity map. We therefore assume that $h'''\neq 0$, and we
introduce the local functional 
\beq\label{ham}
H_{h}:=\int D_{(c,p,s)}(h)\,dx,
\eeq
the \emph{Hamiltonian} of our system. The corresponding evolution equation, that is, the fifth order PDE given by
$$u_{t}=\left\{u,H\right\}=\de_{x}\left(\frac{\delta H_{h}}{\delta u}\right)$$
is a \emph{Hamiltonian perturbation} of the quasilinear PDE \eqref{eq:hopfgen}. Explicitly, we have:
\begin{align}
u_{t}&=h''\,u_{x}+\e^{2}\,\de_{x}\!\left(\frac{1}{2}c_{h}'u_{x}^{2}+c_{h}\,u_{xx}\right)+\e^{4}\de_{x}\bigg(2p_{h}u_{xxxx}\notag\\
&+4p'_{h}u_{x}u_{xxx}+2(p''_{h}-6s_{h})u_{x}^{2}u_{xx}+3p'_{h}u_{xx}^{2}-3s'_{h}u_{x}^{4}\bigg).\label{eveq}
\end{align}
\begin{pro}[Dubrovin, \cite{du06}] The PDE \eqref{eveq} admits infinitely many approximated conserved quantities (up to order $\e^{6}$), of the form
\beq\label{def:cq}
H_{f}=\int D_{(c,p,s)}(f)\,dx,
\eeq 
parametrized by an arbitrary function of one variable. 
\end{pro}
The above proposition can also be stated in the following way: for every $f,\,g\in \mathcal{C}^{\infty}(\mathbb{R})$ the corresponding
functionals $H_{f}$, $H_{g}$ satisfy:
$$\mathcal{E}\left(\frac{\delta H_{f}}{\delta u}\frac{d}{dx}\frac{\delta H_{g}}{\delta u}\right)=O(\e^{6}),$$
where $\mathcal{E}$ is the Euler-Lagrange operator defined in \eqref{eulerlagrange}. 

\begin{exa}
The KdV equation,
$$u_{t}=u\,u_{x}+\e^{2}\,u_{xxx},$$
can be obtained from the above class by choosing $h(u)=\frac{u^{3}}{6}$, $c(u)=1$ and $p(u)=s(u)=0$.  Specializing to
this case, we obtain that the generic $O(\e^{4})-$conserved quantity for KdV is given by \eqref{def:cq} with
$$D_{KdV}(f)=f-\frac{\e^2}{2}\,f'''\,u_{x}^2+
\e^4\left(\frac{3}{10}\,f^{(4)}\, u_{xx}^2-\frac{1}{24}f^{(6)}\,u_{x}^4\right)$$
and $f(u)$ an arbitrary function.
\end{exa}
The functions $c$, $p$ and $s$ appearing above characterize a hierarchy of Hamiltonian PDEs, which commute up to order $\e^{6}$.
Every equation of the hierarchy is generated by an Hamiltonian \eqref{ham} for a specific choice of $h$.  As proved by Dubrovin in
\cite{du06}, the functions $c$ and $p$ (but not $s$) are invariants of the hierarchy with respect to Miura-type transformations of the form
\beq\label{miurapoisson}
u\longmapsto w=u+\sum_{k\geq 1} \e^{k}\,F_{k}(u;u_{x},\dots,u^{(k)}),
\eeq
with $F_{k}$ polynomials in the derivatives of $u$, of total degree $k$.  Equivalently, the class of Hamiltonian perturbations
modulo Miura-type transformations \eqref{miurapoisson} is parametri-zed -- up to order $\e^{4}$ -- by three arbitrary functions of one variable:
$h$, $c$ and $p$, while the function $s$ can be chosen at convenience.  In this paper, we present our the results for an arbitrary choice of the function $s$. Indeed, as shown
for instance in the example below, many important examples of Hamiltonian perturbations are obtained selecting different choices of $s$.

\begin{exa}\label{ourexample}
The following class of generalized KdV equations
$$u_{t}=a(u)\,u_{x}+\e^{2}\,\kappa_{1}\,u_{xxx}+\e^{4}\,\kappa_{2}\,u_{xxxxx},$$
where $a$ is an arbitrary smooth function, $\kappa_{1}$, $\kappa_{2}$ arbitrary constants, can be written as Hamiltonian perturbations by choosing
\begin{gather*}
h''=a,\quad c=\frac{\kappa_{1}}{a'},\quad
p=\frac{\kappa_{2}}{2\,a'}-\frac{3\,\kappa_{1}^{2}}{10}\frac{a''}{(a')^{3}},\\
\nonumber
s=\kappa_{1}^{2}\,\left(\frac{2}{5}\,\frac{(a'')^{3}}{(a')^{5}} -\frac
{7}{20}\frac{a''\,a'''}{(a')^{4}}
+\frac{1}{24}\frac{a''''}{(a')^{3}}\right)
-\frac{\kappa_{2}}{12}\left(\frac{(a'')^{2}}{(a')^{3}}-\frac{a'''}{(a')^{2}}\right).
\end{gather*}
\end{exa}

\subsection{Solutions to the perturbed equation for small times}
Let us consider now solutions of a perturbed equation \eqref{eveq}. The discussion here will be at a formal level,
for in general (that is, for an arbitrary perturbation)  it is not known whether the equation admits a global -- or even local --
solution for a given initial data.  We consider the Cauchy problem
\beq\label{cauchypert}
u_{t}=\left\{u,H_{h}\right\},\qquad u_{|t=0}=\varphi,
\eeq
where the Hamiltonian $H_{h}$ is given by \eqref{ham}, and the initial data $\varphi$, which is $\e$-independent, is the same as
for the unperturbed equation \eqref{eq:hopfgen}. A standard method to obtain approximate solutions to \eqref{cauchypert} is by considering a semiclassical expansion in the small parameter $\e$, which in our case reads as follows: consider the Ansatz
\begin{equation}\label{expu}
u(x,t,\e)=\sum_{i\geq 0} v^{i}(x,t)\,\e^{\,2i},
\end{equation}
with coefficients  $v^{i}$ smooth functions of $x$ and $t$. Within this setting, all identities are understood in the sense of
\emph{formal power series} in $\e$ - they are assumed to hold identically at every order in $\e$.  In particular, by evaluating
\eqref{expu} at $t=0$, one gets the relations
$$v^{0}(x,0)=\varphi(x),\qquad v^{i}(x,0)=0,\,\, i\geq 1.$$
By expanding both sides of equation \eqref{eveq} according
to the Ansatz \eqref{expu}, in first approximation one obtains
\begin{equation}\label{eq:vhopf} 
v^{0}_{t}=h''(v^{0})\,v^{0}_{x},\qquad u_{|t=0}=\varphi,
\end{equation}
which says that $v^{0}$ must be a solution of the Cauchy problem \eqref{eq:hopfgen}.  Accordingly, from the higher order coefficients one obtains an infinite set of \emph{semilinear} equations (or \emph{transport equations}) for the coefficients $v^{k}(x,t)$. For instance, the equation for $v^{1}$ turns out to be
\begin{equation}\label{eq:tr1}
v^{1}_{t}=\partial_{x}\!\left(h''(v^{0}) \,v^{1}+c_{h}(v^{0})\, v^{0}_{xx}+\!\frac{1}{2}c'_{h}(v^{0}) \left(v^{0}_{x}\right) ^{2}\right), \quad v^1|_{t=0}=0,
\end{equation}
while the equation for $v^{2}$ is
\begin{align}
v^{2}_{t}=\partial_{x}&\bigg(h'' v^{2}+\frac{1}{2}h'''(v^{1})^{2}+\de_{x}\left(c_{h}v^{1}_{x}\right)+c'_{h}v^{0}_{xx}v^{1}+\frac{1}{2}c''_{h}(v^{0}_{x})^{2}v^{1}+2p_{h}v^{0}_{xxxx}\notag\\
&+4p'_{h}v^{0}_{x}v^{0}_{xxx}+2p''_{h}(v^{0}_{x})^{2}v^{0}_{xx}+3p'_{h}(v^{0}_{x})^{3}-12s_{h}(v^{0}_{x})^{2}v^{0}_{xx}-3s'_{h}(v^{0}_{x})^{4}\bigg), \label{eq:tr2}
\end{align}
with the initial condition $v^2|_{t=0}=0$. Here, the functions $h$, $c_{h}$, $p_{h}$, $s_{h}$ and their derivatives are evaluated at $u=v^{0}(x,t)$. Note that, although in principle every transport equation can be solved recursively
starting from the solution of \eqref{eq:vhopf}, the use of a standard method for finding the solution requires a case-by-case approach, which can be very difficult to implement, even in simple cases.
\begin{rem}
It is reasonable to expect that expansion \eqref{expu} holds true in the space-time region where the solution of the Hopf equation (\ref{eq:vhopf})
obtained by the method of characteristics is single-valued. This region can be described as the complement of the so-called \lq Whitham zone' \cite{whi74,avno87}. For the class of perturbations of Example \ref{ourexample}, and for times smaller that the critical time, we recently proved in \cite{mara11} that an expansion of the form \eqref{expu} actually exists at some order in $\e$, provided the initial data is sufficiently regular. In the full complement of the Whitham zone, namely also for times $t\geq t_c$, the existence of expansion \eqref{expu} was proved just for special classes of solutions of the KdV equation and its hierarchy - see \cite{clgr11} for a review. 
\end{rem}

%

An alternative way to study solutions of \eqref{cauchypert} for small times has been proposed within the theory of Hamiltonian perturbations in \cite{du06, lizh06}. The idea is to look for a specific \emph{(quasi-)Miura} transformation, mapping any solution of the unperturbed equation, say $v^{0}(x,t)$,  to a solution $u(x,t)$ of the perturbed equation. Remarkably, this is always possible for a $\e^{4}-$perturbation: the transformation is of the form 
\begin{equation}\label{cantrans}
v^{0}\longmapsto u=v^{0}+\e\,\left\{v^{0}(x),K\right\}+ \e^2\left\{\left\{v^{0}(x),K\right\},K\right\}+\dots
\end{equation}
where the functional $K$, up to order $4$ in $\e$, is given by
\begin{equation}\label{borisfunctional}
\!\!\!\!\!K=-\!\!\int\!\!\left[\e\frac{c(v^{0})}{2}\,v^{0}_{x}\log{v^{0}_{x}}+\e^{3}\!\!\left(\frac{c(v^{0})^{2}}{40}\!
\left(\frac{v^{0}_{xx}}{v^{0}_{x}}\right)^{\!\!3}\!\!-\frac{p(v^{0})}{4}\frac{\left(v^{0}_{xx}\right)^{\!2}}{v^{0}_{x}}\right)\!\right]\!\!dx.
\end{equation}
Although this transformation turns out to be very useful when considering algebraic and geometric aspects of the theory of Hamiltonian perturbations
 \cite{du06,liwuzh08}, if one is interested in (approximate) solutions of  \eqref{eveq}, then the above formula has some drawback.
Indeed, it is clear that  (\ref{cantrans}, \ref{borisfunctional}) is singular at critical points of solutions to the equation \eqref{eq:hopfgen};
moreover, the required initial condition $u_{|t=0}=\varphi$ is not satisfied.

In order to eliminate these problems, we construct in the next section a suitable deformation, to the case of Hamiltonian perturbations,
of the characteristic equation \eqref{eq:func}, and we show how to use it in order to explicitly compute the correct terms of the semiclassical expansion.

\section{The variational string equation}\label{sec:string}
The aim of the present section is to construct an analogue of equation \eqref{eq:func} for Hamiltonian perturbations of type \eqref{eveq}. The main idea we follow is to suitably generalize the variational formulation of the method of characteristics, outlined in Section \ref{sec:dztheory}, to the perturbed case. Remarkably, it turns out that this is always possible at order $\e^{2}$, while at order $\e^{4}$ nontrivial constraints appear. There are three properties that we require in order to obtain a meaningful deformation: first, the new equation has to be written as the critical point of an (almost) conserved quantity for the perturbed equation. Then, at $\e=0$ we want to recover  the method of characteristics, and finally, in the limit as $t$ tends to zero 
we want to get the correct initial data.  If $\varphi$ is the initial data of \eqref{eveq}, $g$ a function such that $g'(\varphi(x))=x$, and $S$ denotes the functional we want to find, then these conditions can be written as
\begin{subequations}\label{2st}
\begin{align}
&\partial_t \left(\du{S}\right) + \mathcal{E}\left(\du{S}\frac{d}{dx} \du{H_{h}}\right) = O(\e^{K}),  \qquad \mbox{ for some } K >0, \label{2stcomm}\\
& \du{S}_{|\e=0} = 0 \quad  \Longleftrightarrow \quad  x+ h''(u)\,t = g'(u), \label{2ste0} \\ 
&\du{S}_{|t=0, u=\varphi(x)} = 0.  \label{2stt0} 
\end{align}
\end{subequations}
Maybe only equation \eqref{2stcomm} requires some comment. If $S$ is an exact conserved quantity, explicitly depending on $t$, for equation \eqref{eveq}, then  the right hand side of \eqref{2stcomm} is exactly equal to zero. Since in this case we expect to obtain only an approximate conserved quantity, then \eqref{2stcomm} tells us that the critical locus of $S$ varies, at least formally, slowly: we can expect the solution of the Cauchy problem to satisfy the critical point equation up to an error $O(\e^K)$, uniformly in $x,t$ on some space-time domain.

Given a Cauchy problem for a perturbation \eqref{eveq}, if a special conserved quantity $S$ satisfying conditions \eqref{2st} is found, at some order in $\e$, then the required analogue of the characteristic equation \eqref{eq:func} is given by
\beq\label{se}
\frac{\delta S}{\delta u}=0.
\eeq
We call this equation \emph{variational string equation}. In the special case when the perturbed equation \eqref{eveq} is given by the choice $c(u)=c_{0}, p(u)=p_{0}, s(u)=0$, with $c_{0}$ and $p_{0}$ constants, a possible candidate appears in the Dubrovin's paper \cite{du06}, where he considered the functional
\beq\label{borisconstfunc}
S=\int\Big(x\,u + t\,D_{c_{0},p_{0},0}(h')-D_{c_{0},p_{0},0}(g)\Big) dx,
\eeq
in relation with the universality conjecture of solutions at the critical point.
Note, however, that the above functional satisfies condition \eqref{2stt0} at order $O(\e^{2})$ only. 

In the following, we generalize the functional \eqref{borisconstfunc} to a more general class of perturbations. Note that this quantity is the sum of two separate conserved quantities:
\beq\label{s0const}
S_{0}=\int\Big(x\,u + t\,D_{c_{0},p_{0},0}(h')\Big) dx,
\eeq
and
$$H_{g}=\int D_{c_{0},p_{0},0}(g) dx.$$
We will first consider an extension of the former, requiring only \eqref{2stcomm} and \eqref{2ste0} to be fulfilled.
An appropriate choice of the second conserved quantity will then provide the correct initial data \eqref{2stt0}.

\begin{rem} The reason of the name \lq variational string equation\rq\, is due to the fact that in a particular case identity 
\eqref{se} reduces to the \lq string equation\rq\, associated to the KdV hierarchy \cite{no90,mo90}.
Let us consider the critical point of the following KdV exact conserved quantity
$$
S=\int \left(xu +t\frac{u^2}{2}\right) dx - H_n,
$$
where $H_{n}$ is the n-th Hamiltonian of KdV hierarchy (the deformation of  $\int u^n dx$,  conserved
quantity of the Hopf equation).
The variational string equation for this functional reads
$$
x+ut=\du{H_n} \; .
$$
Note that this is an exact identity. Differentiating the last equation with respect to $x$, we get
$$
1=-t\,u_{x}+\frac{d}{dx}\du{H_n}=-\left[ L,P\right], \qquad P=-(L^{\frac{2n+1}{2}})_+ + t\, (L^{\frac{1}{2}})_+
$$
where $L=-\e^2\frac{d^2}{dx^2}+u$ is the Lax operator of the KdV hierarchy, and $(L^{\frac{2n+1}{2}})_+$ is the differential part
of the pseudo-differential operator $L^{\frac{2n+1}{2}}$. The equation $[P,L]=1$ (after a suitable rescaling $x \to \frac{x}{\e}$) is
known in the literature as string equation.
\end{rem}

\subsection{A special conserved quantity}\label{sub:special}
Consider a family of functionals, depending on $t$, of the form
\beq\label{stringcq}
\int \Big(x\,\rho(u; u_{x},u_{xx},\dots,\e)+t\,\mu(u; u_{x},u_{xx},\dots,\e)\Big)\,dx,
\eeq
for certain differential polynomials
\begin{align*}
\rho=&\rho_{0}+\e^{2}\left(\rho_{1}u_{x}^{2}+\rho_{2}u_{xx}\right)+\e^{4}\left(\rho_{3}u_{x}^{4}+\rho_{4}u_{xx}^{2}+\rho_{5}u_{x}^{2}u_{xx}+\rho_{6}u_{x}u_{xxx}+\rho_{7}u_{xxxx}\right),\\
\mu=&\mu_{0}+\e^{2}\left(\mu_{1}u_{x}^{2}+\mu_{2}u_{xx}\right)+\e^{4}\left(\mu_{3}u_{x}^{4}+\mu_{4}u_{x}^{2}u_{xx}+\mu_{5}u_{xx}^{2}+\mu_{6}u_{x}u_{xxx}+\mu_{7}u_{xxxx}\right),
\end{align*}
where the $\rho_{i}$ and $\mu_{i}$ are arbitrary functions of $u$.  We want to determine conditions on these functions in order to obtain a conserved quantity -- up to order $O(\e^{6})$ -- for equation \eqref{eveq}. As a first step, we  write this functional in normal form, that is up to exact derivatives:
\begin{lem}
Every functional of type \eqref{stringcq} can be written in a unique way in the form 
\beq\label{stringcqnor}
\int \Big(x\,\alpha(u; u_{x},\dots,\e)\,+\e^{4}\,e(u)\,u_{x}^{3}+t\,\beta(u; u_{x},\dots,\e)\Big)\,dx,
\eeq
where
\begin{align*}
\alpha=&\alpha_{0}+\e^{2}\alpha_{1}\,u_{x}^{2}+\e^{4}\left(\alpha_{2}\,u_{x}^{4}+\alpha_{3}\,u_{xx}^{2}\right),\\
\beta=&\beta_{0}+\e^{2}\beta_{1}\,u_{x}^{2}+\e^{4}\left(\beta_{2}\,u_{x}^{4}+\beta_{3}\,u_{xx}^{2}\right),
\end{align*}
and $\alpha_{i}(u)$, $\beta_{i}(u)$ and $e(u)$ are functions of one variable.
\end{lem}
\begin{proof} A direct computation, using integration by parts.
\end{proof}
\begin{pro}\label{specialconserved2}
For arbitrary functions $h, c, p$ and $s$, with $c$ not identically zero, equation \eqref{eveq} admits a conserved quantity of order $O(\e^{2})$ of type \eqref{stringcqnor}.
This is specified by choosing $\alpha=D_{(c,p,s)}(r)$ and $\beta=D_{(c,p,s)}(q)$ in  \eqref{stringcqnor}, where the functions $r(u)$ and $q(u)$ satisfy
\beq\label{rcqhc}
r'=\frac{1}{c},\qquad q'=\frac{h''}{c}.
\eeq
The $O(\e^{2})$-conserved quantity is unique up to a multiplicative constant.
\end{pro}
\begin{proof}
A direct calculation: consider a functional of type \eqref{stringcqnor} and impose that it is conserved at order $O(\e^{2})$
with respect to equation \eqref{eveq}. Denote $\alpha_{0}=r, \,\beta_{0}=q$.  Then, at order $O(1)$, one gets the relation
\beq\label{e0proof}
r'=q'\,h'',
\eeq
while at order $O(\e^{2})$ there are three conditions: the first two, 
$$\alpha_{1}=-\frac{1}{2}c\,r''',\qquad \beta_{1}=-\frac{1}{2}c\,q''',$$  
imply that -- up to order $\e^{2}$ -- we have $\alpha=D_{(c,p,s)}(r)$ and $\beta=D_{(c,p,s)}(q)$. The last condition is given by
$$3r''ch'''+r''' c h'' +r' c h^{(4)}+r' c' h'''-cq'''=0,$$
and after substituting \eqref{e0proof} into it, one obtains \eqref{rcqhc}.
\end{proof}
Let us consider now the $O(\e^4)$ term.
\begin{pro}\label{specialconserved4}
For arbitrary functions $h$, $c$ and $s$, with $c$ not identically zero, equation \eqref{eveq} admits a conserved quantity of order $O(\e^{4})$ of type \eqref{stringcqnor} if and only if 
\beq\label{rel:cp}
p(u)=\frac{3}{5}c(u)\,c'(u)+\lambda \,c(u)^{3},
\eeq
where $\lambda$ is an arbitrary constant. In this case, we have 
\begin{align}
e=&\frac{53}{24}\,\frac{c''\left(c'\right)^{2}}{c^{2}}-\frac{8}{15}\,\frac{ \left(c'\right)^{4}}{c^{3}}+\frac {s'}{c}-\frac{157}{120}\,\frac{\left(c''\right)^{2}}{c}-5\,\frac {s\,c'}{c^{2}}-\frac {173}{120}\,\frac{c'\,c'''}{c}+\frac {5}{12}c^{(4)}\notag\\
&+\lambda\left(-\frac{13}{2}c''c'+\frac{2\,\left(c'\right)^{3}}{c}+\frac{11}{6}c\,c'''\right).\label{gauge:e}
\end{align}
The $O(\e^{4})$-conserved quantity of type \eqref{stringcqnor} is given as in Proposition \ref{specialconserved2}, and it is unique up to a multiplicative constant.
\end{pro}
\begin{proof} As for Proposition \ref{specialconserved2}, the condition of \eqref{stringcqnor} being a conserved quantity at order $\e^{4}$ implies that $\alpha$ and $\beta$ are themselves $O(\e^{4})$-conserved quantities. In addition, one finds two more conditions. One of them involves the function $e(u)$ and is satisfied by the choice \eqref{gauge:e}, the other reads
$$\frac{h'''}{c^{2}}\left(p'c-3\,p\,c'+\frac{6}{5}c\,(c')^{2}-\frac{3}{5}\,c^{2}c''\right)=0,$$
from which one obtains relation \eqref{rel:cp}.
\end{proof}

Due to the above propositions, the required functional has the form
\beq\label{s0}
S_{0}=\int\Big(x\,D_{c,p,s}(r)+t\,D_{(c,p,s)}(q)+\e^{4}\,e(u)\,u_{x}^{3}\Big)\,dx, 
\eeq
with $r$ and $q$ given by \eqref{rcqhc} and $e$ given by \eqref{gauge:e}. By construction, condition \eqref{2stcomm} is satisfied for any Hamiltonian perturbation at order $\e^{2}$, and for the class of equations characterized by \eqref{rel:cp} at order $\e^{4}$. Note that since the density of $S_{0}$ depends explicitly on $x$, its variational derivative is given by
$$\frac{\delta S_{0}}{\delta u}= x\mathcal{E}(D_{c,p,s}(r))+\sum_{k\geq 0}(-1)^{k+1}\frac{d^{k}}{d x^{k}}\frac{\de D_{c,p,s}(r)}{\de u^{(k+1)}} +\mathcal{E}\Big(t\,D_{(c,p,s)}(q)+\e^{4}\,e(u)\,u_{x}^{3}\Big).$$

\begin{exa} In the particular case $c(u)=c_{0}, p(u)=p_{0}, s(u)=0$, then the functional \eqref{s0} reduces to \eqref{s0const}. Note that \eqref{rel:cp} is satisfied by choosing $\lambda=p_{0}/c_{0}^{3}$.
\end{exa}

\begin{rem}
One of the hypothesis on the Hamiltonian perturbation in Proposition \ref{pro:s2} and \ref{pro:s4} is that $c(u)$ does not vanish identically. Indeed, if $c \equiv 0$ then the proof breaks since the special conserved quantity $S_0$ is built on the function $\frac{1}{c}$. Although we do not treat the case $c \equiv 0$ in generality, the following simple example shows that the string equation can be considered also in this case. The generalized KdV equation
$$u_t=u\,u_x+ \e^4 \,   u_{xxxxx} \; ,$$
obtained choosing  $c=0, p=\frac{1}{2}, s=0$ in \eqref{ham}, admits the exact conserved quantity
\beq\label{galilean}
S_0= \int \left(x \, u  + t \,\frac{u^{2}}{2}\right) \, dx \; .
\eeq
Note that  $S_{0}$ is nothing but the conserved quantity associated to the Galilean invariance of the equation.
\end{rem}

\begin{rem}
A natural question arising in the above construction is whether the constraint \eqref{rel:cp} at order $\e^{4}$
can be eliminated by a different choice of the class of conserved quantities
(e.g., depending rationally on the derivatives).  This, however, is outside the scope of the present paper.
\end{rem}

\subsection{Fixing the initial data} \label{sub:initialdata}
We now address the problem of recovering the required initial data, namely to modify \eqref{s0} to get a functional which satisfies \eqref{2stt0}. Consider a Cauchy problem for equation \eqref{eveq}, with given initial data $u(x,0,\e)=\varphi(x)$, independent of $\e$. In addition, we introduce the function $f$, defined by the relation $g'=c\,f'$. Due to the definition of $g$, we thus have that $f$ is locally related to the initial data by the identity
\beq\label{fid}
c(\varphi(x))\,f'(\varphi(x))=x,
\eeq
for $x$ belonging to some interval of the real axis. The correction of the initial data (in this case, at order at most $\e^{4}$) can be obtained as follows: for arbitrary functions $f_{1}$ and $f_{2}$, we consider the functional
$$H_{f+\e^{2}f_{1}+\e^{4}f_{2}}=\int D_{c,p,s} (f+\e^2 f_1+\e^4 f_2) dx, $$
and then truncate it at order $\e^{4}$, by introducing the $4^{th}$ order (in $\e$) polynomial $H_{f}^{\e}$, defined by the relation $H_{f}^{\e}-H_{f+\e^{2}f_{1}+\e^{4}f_{2}}=O(\e^{6})$. Finally, define
\beq\label{eq:sgeneral}
S_f:= S_{0}- H^{\e}_{f},
\eeq
which by construction satisfies \eqref{2stcomm}\footnote{at order $\e^{2}$. The identity at $\e^{4}$ is verified subject to condition \eqref{rel:cp}.} and, due to the definition of $f$, also \eqref{2ste0}. Note that \eqref{eq:sgeneral} itself is a polynomial of order $4$ in $\e$. A suitable choice of the functions $f_{1}$ and $f_{2}$ implies condition \eqref{2stt0}.

\begin{pro}\label{pro:s2}
For any Hamiltonian perturbation of type (\ref{eveq}) with arbitrary coefficients $h,c,p,s$ the functional \eqref{eq:sgeneral}
with $S_{0}$ given by (\ref{stringcq}), is conserved up to $O(\e^4)$.
Moreover, if $f$ is chosen such that \eqref{fid} holds and
\begin{align}\label{kdv:f1}
f'_{1}(v)&=\frac{1}{2\,c(v)}\frac{d^{2}}{dv^{2}}\left(\frac{c(v)}{(c(v)f'(v))'}\right),
\end{align}
then we have
$$\du{S_f}_{|t=0, u=\varphi(x)}=O(\e^4).$$
\end{pro}
\begin{proof}  The fact that the functional \eqref{eq:sgeneral} is a $O(\e^{2})$-conserved quantity of \eqref{eveq} for any choice of the functions $f$ and $f_{1}$ follows immediately from the definition of the linear operator $D_{c,p,s}$. Let us prove the second part of the proposition. The condition
$$\du{S_f}_{|t=0, u=\varphi(x)}=O(\e^2)$$
is verified due to the choice \eqref{fid}. Moreover, the $\e^{2}$ coefficient of $S_{f}$ is given by
\begin{equation*}
-\int \left[f_{1}+\frac{c}{2} \bigg(x\,\left(\frac{1}{c}\right)''+t\,\left(\frac{h''}{c}\right)''-f'''\bigg)u_{x}^{2}\right] dx .
\end{equation*}
The variational derivative of this functional, evaluated at $t=0$, $u=\varphi(x)$ and set equal to zero gives the following condition on $f_{1}$:
\begin{align}
f'_{1}(\varphi(x))=& \,x\left[\left(\frac{2(c')^{2}}{c^{2}}-\frac{c''}{c}\right)\varphi_{xx}(x)+\left(\frac{5}{2}\frac{c'c''}{c^{2}}-2\frac{(c')^{3}}{c^{3}}-\frac{1}{2}\frac{c'''}{c}\right)\varphi_{x}(x)^{2}\right]\notag\\
+&\left(\frac{2(c')^{2}}{c^{2}}-\frac{c''}{c}\right)\varphi_{x}(x)-c\,f'''\varphi_{xx}(x)+\frac{1}{2}\left(c'f'''+cf^{(4)}\right)\varphi_{x}(x)^{2}. \label{f1t0id}
\end{align}
Of course here the functions $f, c$ and their derivatives are computed at $u=\varphi(x)$. Using condition \eqref{fid}, together with the differential identities
$$\varphi_{x}(x)=\frac{1}{(c\,f')'}|_{u=\varphi(x)},\qquad \varphi_{xx}(x)=-\frac{(c\,f')''}{[(c\,f')']^{3}}|_{u=\varphi(x)},$$
obtained from it, one gets that \eqref{f1t0id} is satisfied by the choice \eqref{kdv:f1}. In particular, this implies
$$\du{ S_{f}}_{|t=0,u=\varphi(x)}=O(\e^{4}),$$
and the thesis is proved.
\end{proof}

\begin{pro}\label{pro:s4}
Under the hypotheses of the previous proposition, and  if in addition the coefficient $p$ of the Hamiltonian satisfies condition $\eqref{rel:cp}$, then the functional \eqref{eq:sgeneral} is conserved up to $O(\e^6)$. Moreover, if $f_{2}$ is chosen such that
$$f'_{2}(v)=\frac{1}{c(v)}\frac{d^{2}}{dv^{2}}\hat{f}_{2}(v) \; , $$
where
\begin{align*}
\hat{f}_{2}=&\frac{p}{2\,c}\left[\frac{c}{\left[(c\,f')'\right]^{2}}\frac{d^{2}}{dv^{2}}\left(\frac{1}{(c\,f')'}\right)+
c'\frac{d}{dv}\left(\frac{1}{\left[(c\,f')'\right]^{3}}\right)\right]\frac{s}{\left[(c\,f')'\right]^{3}}\\
&+\frac{1}{6}\frac{c^{2}}{\left[(c\,f')'\right]^{2}}\frac{d^{3}}{dv^{3}}\left(\frac{1}{(c\,f')'}\right)
\frac{5}{24}\frac{c\,c'}{\left[(c\,f')'\right]^{2}}\frac{d^{2}}{dv^{2}}\left(\frac{1}{(c\,f')'}\right)\\
+&\Bigg[\frac{49}{120}c\,c'\,c''f'\frac{1}{10}c^{3}\,f^{(4)}-\frac{1}{10}c^{2}c'''f'-\frac{3}{5}(c')^{3}f'+\frac{11}{60}c^{2}\,c'\,f'''\\
&+\frac{11}{30}c\,(c')^{2}\,f''-\frac{3}{8}c^{2}\,c''\,f''\Bigg]\frac{1}{\left[(c\,f')'\right]^{3}}\frac{d}{dv}\left(\frac{1}{(c\,f')'}\right)\\
+&\Bigg[\frac{1}{2}(c')^{2}f'''+\frac{1}{2}c'\,c''\,f''-\frac{1}{4}c\,c''f'''-\frac{1}{4}c\,c'''f''+\frac{1}{12}c'c'''f'\\
&-\frac{1}{24}c^{2}\,f^{(5)}-\frac{1}{24}c\,c^{(4)}\,f'\Bigg]\frac{c}{\left[(c\,f')'\right]^{2}},
\end{align*}
then we have
$$
\du{S_f}_{|t=0, u=\varphi(x)}=0.
$$
\end{pro}
\begin{proof}
The proof is similar to the previous proposition.
\end{proof}
The functional \eqref{eq:sgeneral} with $f_{1}$ and $f_{2}$ given as in the above propositions is the required functional for the variational string equation, as it satisfies conditions \eqref{2st} at the desired order. The variational string equation (of order $K$) associated to $S_{f}$ is defined as the equation for the critical point: 
\beq\label{stringequation}
\du{S_{f}}=0.
\eeq
Explicitly, at order $O(\e^{2})$, the variational string equation reads
\begin{align*}\label{stringequation}
&\frac{1}{c}\left(x+t\,h''-c\,f'\right)+\e^{2}\Bigg\{x\left[c\,\Big(\frac{1}{c}\Big)''u_{xx}+\frac{1}{2}\Big(c\,\Big(\frac{1}{c}\Big)''\Big)'u_{x}^{2}\right]+c\,\Big(\frac{1}{c}\Big)''u_{x}-c\,f'''u_{xx}\\
&+\frac{1}{2}\left(c\,f'''\right)'u_{x}^{2}-\frac{1}{2\,c}\left(\frac{c}{(c\,f')'}\right)''+t\left[c\,\Big(\frac{h''}{c}\Big)''u_{xx}+\frac{1}{2}\Big(c\,\Big(\frac{h''}{c}\Big)''\Big)'u_{x}^{2}\right]\Bigg\}=O(\e^{4}).
\end{align*}
It is a nontrivial and remarkable fact that this equation provides approximate solutions to a nonlinear PDE for arbitrary initial data.\\

It is natural to ask in which case the construction of the variational string equation can be extended higher orders, and
more generally how to select Hamiltonian perturbations admitting a variational string equation at any order. There are two main
difficulties to solve this problem: first of all, for a generic perturbation and for an arbitrary function $f$, it is not known whether
the functional $\int f(u) dx$ can be deformed at any order to  a conserved quantity with density depending \emph{polynomially} on the jet variables,
as required in the construction of the variational string equation. In the language of \cite{lizh06}, this requirement is the definition
of \emph{formal integrability} for the perturbed equation. In addition, it is in general not known whether special conserved quantities
of type $S_{0}$ exist at any order. A positive answer to these problems is given by the KdV example. Indeed, due to Galilean invariance, KdV
admits the exact conserved quantity \eqref{galilean}, and the equation itself has been proved in \cite{lizh06} to be formally integrable.
Therefore, the string equation in this case can be extended at any order.

\section{Approximate solutions to perturbed equations}\label{sec:solutions}
We now use the variational string equation \eqref{stringequation} to compute the first two perturbative corrections of
the solution $u(x,t,\e)$ to an Hamiltonian perturbation \eqref{eveq}. In \cite{mara11}, we gave the first correction $v^{1}$
without explaining in detail the procedure followed to obtain it. Here we fill this gap, and compute explicitly the second term $v^{2}$.

Suppose the Cauchy problem of (\ref{eveq}) admits the semiclassical expansion \eqref{expu}, and suppose in addition that the
equation admits a variational string equation up to  order $O(\e^{2N})$. Then, we first require $u(x,t,\e)$ to be an (approximate) solution of the variational string equation:
$$\du{S_{f}}_{|u=u(x,t,\e)}=O(\e^{2N+2}),$$
and we expand the above formal identity according to the Ansatz
\eqref{expu}. At order $O(1)$, we find
\beq\label{striue0}
x+t\,h''(v^{0})=c(v^{0})\,f'(v^{0}),
\eeq
which is precisely the formula for the method of characteristics, thus implying that $v^{0}$ is a solution of the dispersionless equation, as required. Consider now the $\e^{2}$ term. This reads
\begin{align*}
&\left(\frac{x\,c'+t\,(h''c'-h'''c)+c^{2}\,f''}{c^{2}}\right) v^{1}=x\left[c\,\Big(\frac{1}{c}\Big)''v^{0}_{xx}+\Big(c\,\Big(\frac{1}{c}\Big)''\Big)'\frac{(v^{0}_{x})^{2}}{2}\right]+c\,\Big(\frac{1}{c}\Big)''v^{0}_{x}\\
&-c\,f'''v^{0}_{xx}+\left(c\,f'''\right)'\frac{(v^{0}_{x})^{2}}{2}-\frac{1}{2\,c}\left(\frac{c}{(c\,f')'}\right)''+t\left[c\,\Big(\frac{h''}{c}\Big)''v^{0}_{xx}+\Big(c\,\Big(\frac{h''}{c}\Big)''\Big)'\frac{(v^{0}_{x})^{2}}{2}\right],
\end{align*}
where all functions are evaluated at $u=v^{0}(x,t)$. The explicit dependence on $f$ of the above formula can easily be eliminated:
further differentiating the exact identity \eqref{striue0} with respect to $x$, one obtains a formula for the quantities $f^{(i)}(v^{0})$.
Substituting into the equation for $v^{1}$, this gives
\begin{equation}\label{eq:v1}
 v^1(x,t)=\frac{t}{2}\,\frac{\partial}{\partial x}\! \left(\frac{\left(c\,h'''\right)'
(v^{0}_{x})^{2}+2\,c\,h'''\,v^{0}_{xx}+t\,c\,(h''')^{2}v^{0}_{x}\,v^{0}_{xx}+t\,c'\,(h''')^{2}
(v^{0}_{x})^{3}}{(1+t\,h'''\,v^{0}_{x})^{2}}\right),
\end{equation}
which is a universal formula: given the solution to the dispersionless Cauchy problem \eqref{eq:vhopf}, then \eqref{eq:v1} satisfies the
transport equation \eqref{eq:tr1} with the correct initial data. Note that using the variational string equation  the problem is solved by a
finite number of elementary manipulations. A more elegant representation of $v^{1}$, first introduced in \cite{mara11}, is given by considering
the functional (depending on $t$):
\beq\label{k1}
K_{1}[u]:=-\frac{1}{2}\int c\,(u)\,u_{x}\,\log{\Big(1+t\,h'''\left(u\right)u_{x}\Big)}\,dx.
\eeq
Then, the function $v^{1}$ admits the Hamiltonian representation
$$v^{1}(x,t)=\frac{\partial}{\partial x}\left(\frac{\delta K_{1}[u]}{\delta u(x)}_{|u=v^{0}(x,t)}\right).$$
The calculation of the higher order corrections can be computed similarly, provided the variational string equation associated with the
PDE exists at the desired order. For instance, let us consider the second correction term for the Hamiltonian satisfying condition \eqref{rel:cp}.
Since the formulae in this case are much more cumbersome, we present here only the more compact variational formulation. Indeed, it turns out that
the transformation from $v^0(x,t) \to u(x,t,\e)$ can be written -- up to order $\e^{4}$ --  in the canonical form
\begin{align}\label{eq:canonicK}
u &=v^0+\e\, \lbrace v^0, K \rbrace +\frac{\e^2}{2} \lbrace \lbrace v^0, K \rbrace , K\rbrace +\dots, \\
K &= \e K_1+\e^3 K_3 +\dots ,  \nonumber
\end{align}
where $K_{1}$ is given by \eqref{k1} and $K_{3}$ is the following expression:
\begin{align*}\label{eq:k3}
K_{3}[u]=&-\left(\left(\frac{c}{(h''')^{2}}\right)'''+\frac{c'''}{(h''')^{2}}\right) \frac{c\,u_{x}\log{(1+t\,h'''u_{x})}}{8\,t^{2}}\\ \nonumber
&-\frac{1}{40} \frac{c^{2}(h''')^{3}t^{3} u_{xx}^{3}\left(5+5\,t\,h'''\,u_{x}+t^{2}\,(h''')^{2}\,u_{x}^{2}\right)}{(1+t\,h'''u_{x})^{5}}\\ \nonumber
&+\frac{3}{40}\frac{c^{2}\,t\,h^{(4)}u_{xx}^{2}\left(4+10\,t\,h''' u_{x}+5\,t^{2}\,(h''')^{2} u_{x}^{2}\right)}{(1+t\,h'''u_{x})^{5}}+\frac{1}{2}\,\frac{t\,h''' s\,u_{x}^{4}(2+t\,h'''u_{x})}{(1+t\,h'''u_{x})^{2}}\\ \nonumber
&+\frac{1}{20}\frac{(3cc'+5\lambda c^{3})\,t\,h'''u_{xx}^{2}(2+t\,h''' u_{x})\left(2+2\,t\,h''' u_{x}+t^{2}\,(h''')^{2} u_{x}^{2}\right)}{(1+t\,h'''u_{x})^{4}}\\ \nonumber
&-t\,u_{x}^{4}\left(\frac{c'}{20}+\frac{\lambda}{12}c^{2}\right)\left(\frac{2\,c\,h^{(5)}+6\,c'\,h^{(4)}}{(1+t\,h'''u_{x})^{3}}-\frac{3\,t\,c\,(h^{(4)})^{2}u_{x}}{(1+t\,h'''u_{x})^{4}}\right)\\ \nonumber
&+\frac{u_{x}^{2}}{240\,t}\Bigg[\frac{12\,c^{2}\,(h^{(4)})^{3}}{(h''')^{4}\,(1+t\,h'''u_{x})^{5}}+\frac{3\,c\,h^{(4)}\left(3\,h'''\,h^{(4)}c'+3\,h^{(5)}h'''c-16\,c(h^{(4)})^{2}\right)}{(h''')^{4}\,(1+t\,h'''u_{x})^{4}}\\ \nonumber
&-\frac{3h^{(4)}\left(8\,c\,c''(h''')^{2}-4\,c^{2}(h^{(4)})^{2}+9\,c\,c'\,h'''\,h^{(4)}+9\,c^{2}\,h'''\,h^{(5)}-16\,(c')^{2}\,(h''')^{2}\right)}{(h''')^{4}(1+t\,h'''\,u_{x})^{3}}\\
\nonumber
& -\frac{96\,(c')^{2}(h''')^{3}h^{(4)}-48\,c\,c''(h''')^{2}h^{(4)}+63\,c\,c'\,h'''(h^{(4)})^{2}+5\,c\,c'''(h''')^{3}}{(h''')^{4}(1+t\,h'''\,u_{x})^{2}}\\
\nonumber
& +\frac{252\,c^{2}(h^{(4)})^{3}-63\,c^{2}h'''h^{(4)}h^{5}}{(h''')^{4}(1+t\,h'''\,u_{x})^{2}} +\frac{-30\,c^{2}\,(h''')^{2}h^{(6)}+441\,c^{2}\,h'''\,h^{(4)}\,h^{(5)}}{(h''')^{4}(1+t\,h'''\,u_{x})}
\end{align*}
\begin{align*}
\nonumber
& -\frac{90\,c\,c'(h''')^{2}h^{(5)}-441\,c\,c'h'''\,(h^{(4)})^{2}+768\,c^{2}\,(h^{(4)})^{3}+114\,c\,c''\,(h''')^{2}h^{(4)}}{(h''')^{4}(1+t\,h'''\,u_{x})}\\ \nonumber
&+\frac{48\,(c')^{2}\,(h''')^{2}h^{(4)}+35\,c\,c'''\,(h''')^{3}}{(h''')^{4}(1+t\,h'''\,u_{x})} +\frac{180\,c\,h^{(4)}}{(h''')^{3}}\left(c\,h^{(4)}\right)'-\frac{90\,c}{(h''')^{2}}\left(c'\,h^{(4)}\right)'\\ \nonumber
& -\frac{30\,c^{2}\,h^{(6)}}{(h''')^{2}}-\frac{180\,c^{2}\,(h^{(4)})^{3}}{(h''')^{4}}+\frac{35\,c\,c'''}{h'''}-\frac{5\,c\,c'}{h'''}(1+t\,h'''\,u_{x})\Bigg].
\end{align*}
The canonical transformation (\ref{eq:canonicK}), with $K_1$ and $K_3$ as above, is well-defined for the solution $v^0$ of
the unperturbed equation up to the time of the gradient catastrophe. The apparent singularities of $K_3$ at $t=0$ and $h'''(u)=0$ cancel out; in fact, $K_3$, as well as $K_1$, vanishes identically if $t=0$ or $h'''=0$. In the particular case when $h''(u)=u$, the above formula reduces to
\begin{align*}
K_{3}[u]=& -\frac{1}{4}\frac{c\,c'''u_{x}\log{(1+t\,u_{x})}}{t^{2}}-\frac{1}{40}\frac{t^{3}c^{2}(u_{xx})^{3}(5+5\,t\,u_{x}+t^{2}(u_{x})^{2})}{(1+t u_{x})^{5}}\\
&+\frac{1}{20}\frac{t\,c\,(3c'+5c^{2}\lambda)(u_{xx})^{2}(2+t\,u_{x})(2+2\,t\,u_{x}+t^{2}(u_{x})^{2})}{(1+t\,u_{x})^{4}}\\
&+\frac{1}{48}\frac{(2+t\,u_{x})\left[c\,c'''(6+6\,t\,u_{x}-t^{2}(u_{x})^{2})+24\,s\,t\,(u_{x})^{2}\right]}{t(1+t\,u_{x})^{2}}.
\end{align*}
A direct computation shows that the corrections $v^1$, $v^2$ computed from the canonical transformation (\ref{eq:canonicK})
satisfy the appropriate transport equations (\ref{eq:tr1}, \ref{eq:tr2}). This is an \emph{a posteriori} confirmation of the validity of the variational string equation.

\begin{exa} For $t\in[0,t_{c})$, the solution to KdV with initial data $\varphi$ is given, at order $\e^{4}$, by
\begin{align*}
u=v^{0}&+\frac{\e^{2}}{2}\de_{x}^{2}\left(\log{(1+t\,v^{0}_{x})}-\frac{1}{(1+t\,v^{0}_{x})}\right)+\frac{\e^{4}\,t^{2}}{8}\de_{x}^{2}\Bigg[-\frac{t^{2}}{10}\frac{(5+t\,v^{0}_{x})(v^{0}_{xx})^{3}}{(1+t\,v^{0}_{x})^{5}}\\
& +t\,\de_{x}\left(\frac{3}{10}\frac{20+15\,t\,v^{0}_{x}+3\,t^{2}(v^{0}_{x})^{2}}{(1+t\,v^{0}_{x})^{5}}\right)+\de_{x}^{2}\left(\frac{(2+t\,v^{0}_{x})^{2}}{(1+t\,v^{0}_{x})^{4}}\,v^{0}_{xx}\right)\Bigg] +O(\e^{6}),
\end{align*}
where $v^{0}$ is the solution of Hopf with same initial data. Note that at $t=0$, the above formula gives $u(x,0,\e)=v^{0}(x,0)=\varphi(x)$.
\end{exa}

\section{The string equation is an approximate identity}\label{sec:validity}
In this Section we discuss in which precise sense solutions of a
Hamiltonian perturbation of the Hopf equation satisfy a variational string
equation. Indeed, given the Cauchy problem (\ref{cauchypert}) and the functional $S_f$ conserved up to $O(\e^K)$,
we will prove that the solution $u(x,t,\e)$ of the Cauchy problem approximatively satisfies the string equation: 
$$\du{S_f}(u(x,t,\e))=O(\e^K),$$
provided it is sufficiently regular in the $\e=0$ limit. The problem presents many difficulties, as the string equation is the critical point equation for an ill-defined functional (in the generic case, $g'(0) \neq 0$, the integral diverges for any rapidly decreasing function), which is not exactly conserved. Moreover, the string equation itself is singular and valid only locally in the space-time plane.

The first step of our proof is to show that the variational derivatives of a formal functional satisfies
a \emph{linear} PDE with a forcing term, which vanishes in case the functional is (formally) conserved.
In the case of the string equation, this forcing term will be small, proportional to  $\e^{K}$ for some $K\geq 1$.

After this result, we are lead to analyze local solutions to linear PDEs with a small forcing term. Here, the standard methods of functional analysis are not effective,
as the equation is valid on some bounded domain of the real line and we ignore the boundary values. Equivalently, when estimating the time derivative of the
$L^2$ norm of the solution and integrating by parts some unknown boundary terms appear.

To overcome this problem, we will estimate the (time derivative of the) $L^2$ norm on an interval whose extremes vary with time.
Remarkably, if the extremes vary with a simple
law - reminiscent of the characteristics' equation -  the boundary terms cancel-out.

\subsection{A linear equation for conserved quantities}
Here we prove that the variational derivative $\du{M}$ of a formal functional $M$ satisfies
a linear equation with a forcing term which vanishes in case the functional is formally conserved.
\begin{thm}\label{thm:raymondGelfand}
Let $M$ be a formal local functional, possibly depending explicitly on $x,t$, let $\tilde{m}=\du{M}$ be its variational derivative,
and let $u(x,t)$ be a solution of the Hamiltonian PDE 
$$u_t=\frac{d}{d x}\du{H}  \; .$$ 
Then, the function of two variables $\mathcal{M}(x,t):=\tilde{m}(x,t,u(x,t),u_{x}(x,t),\dots),$ satisfies the linear equation
\beq\label{eq:raymondGelfand}
\frac{\de\mathcal{M}}{\de t} =\sum_{i\geq 0} \mathcal{A}_i(x,t)
\frac{\de^{i+1}\mathcal{M}}{\de x^{i+1}} + \mathcal{R}(x,t),
\eeq
where
$$\mathcal{A}_i(x,t)= \left(\frac{\partial}{\partial u^{(i)}}\du{H}\right)_{|u=u(x,t)},\quad \mathcal{R}(x,t) = \left( \partial_t \du{M} + \eu{\big(\du{M}\frac{d}{dx}\du{H}\big)} \right)_{|u=u(x,t)},$$
provided all the partial derivatives on the right hand side exist and are continuous. 
\begin{proof}
For any sufficiently regular solution $u$, we have that
\begin{align*}
\frac{\de \mathcal{M}}{\de t} &= \left( \partial_t \du{M} + 
\sum_{i\geq 0} \frac{\partial }{\partial u^{(i)}}\du{M}\, \frac{d u^{(i)}}{dt} \right)_{|u=u(x,t)}, \\
&=  \left( \partial_t \du{M} + \sum_{i\geq 0} \frac{\partial }{\partial u^{(i)}}\du{M}\,  \frac{d^{i+1} }{d x^{i+1}} \du{H}\right)_{|u=u(x,t)}. 
\end{align*}
The thesis follows from the remarkable identity due to Gel'fand and Dikii \cite{gelfand76}:
$$\sum_{i\geq 0} \frac{\partial \tilde{f}}{\partial u^{(i)}} \frac{d^{i+1}}{d x^{i+1}}\tilde{g} - \sum_{i\geq 0}
\frac{\partial \tilde{g}}{\partial u^{(i)}} \frac{d^{i+1}}{d x^{i+1}} \tilde{f}=
\eu{}\big( \tilde{f} \frac{d}{dx} \tilde{g} \big),$$
which is valid for any two $\tilde{f},\tilde{g}$ variational derivatives of formal local functionals $\tilde{f}=\du{F}$, $\tilde{g}=\du{G}$.
\end{proof}
\end{thm}
Let us analyze the result of the previous theorem.  In case $M$ is formally conserved, equation (\ref{eq:raymondGelfand}) holds with $\mathcal{R}=0$
and it therefore reduces to a linear PDE. In this particular case, the equation was already obtained by Lax  in \cite{lax75}.
Theorem \ref{thm:raymondGelfand} above extends Lax' result in two directions: first, we do not assume $M$ to be conserved, secondly, the equation remains valid in case $M$ is just a formal functional.  Such a generalization is important for us, as the variational string equation is the critical point equation of a formal functional $S_f$ that is not strictly conserved.

\subsection{The domain of the variational string equation}
We begin now to give a precise sense to the variational string equation; as a first step, we define the domain where the equation is defined. Let us choose an Hamiltonian perturbation of type \eqref{eveq}, admitting a string equation up to a certain order in $\e$. In addition, we fix:
\begin{itemize}
 \item A solution
$u(x,t;\e)$ of the corresponding Cauchy problem
\begin{equation}\label{eq:pde}
u_t= \partial_x \du{H_h} = h''(u)u_x+ O(\e^2), \qquad u(x,t=0,\e)=\varphi(x).
\end{equation}
\item An interval $(x_-,x_+)$ of strict monotonicity of
$\varphi$  such that $c(\varphi(x)) \neq 0$
for all $x \in (x_-,x_+)$.
\end{itemize}
Let $f': (a,b) \to \bb{R}, $ be defined by the relation
$$f'(\varphi(x))=\frac{x}{c(\varphi(x))}, \qquad x\in (x_-,x_+)  \; .$$
and let $S_f$ be the special quasi-conserved functional constructed in Section \ref{sec:string}. We recall that $S_f$ satisfies the properties \eqref{2st}, which can now be written as:
\begin{subequations}\label{st}
\begin{align}
&\partial_t \du{S_f}+\mathcal{E}\big({\du{S_f}\frac{d}{dx}\du{H_h}} \big) = \e^{2N+2} r \quad, \mbox{ for some } N\geq 1. \label{eq:stcomm}\\
&  \du{S_f}_{|\e=0} = \frac{1}{c(u)}\left( x+ h''(u) t - c(u)f'(u)  \right)\\
&\du{S_f}_{|t=0,u=\varphi(x)}  = 0 ,\qquad x \in (x_-,x_+)  \label{eq:stt0}
\end{align}
\end{subequations}
Here $r $ is some differential expression, polynomial in $\e^2$. 
We can now start to describe the domain of the string equation. Let us denote by $\sigma$, the variational derivative of $S_f$
evaluated at $u=u(x,t;\e)$, that is:
$$\sigma(x,t,\e):=\du{S_f}(u(x,t;\e)) \;.$$
Then, the following facts hold:
\begin{itemize}
\item[i)] The function $\sigma(x,t,\e)$ is well defined whenever $u(x,t,\e)$ belongs to the domain of $f'$, namely the interval $(a,b)$. Indeed,  the only source of singularities is the zero locus of $c(u)$ and, by construction, $c(u) \neq 0$ on $(a,b)$.

\item[ii)] By construction, $\sigma(x,t,\e=0)=0$ on the domain of the $(x,t)$ plane filled by the characteristics line emanating from the interval $(x_-,x_+)$. Indeed, on that domain equation \eqref{eq:func} of the method of characteristics holds.
\item[iii)] Due to \eqref{eq:stt0}, we have $\sigma(x,t=0,\e)=0$ if $x \in (x_-,x_+)$.
\end{itemize}
A meaningful domain for $\sigma$ must contain the subsets described in ii) and iii), where $\sigma$ is known to vanish. In addition, we ask the domain of $\sigma$ to be connected.
\begin{Def}\label{Sigmadomain}
Let $\Omega=\lbrace (x,t,\e) \in \bb{R}^3 || u(x,t,\e) \in (a,b)  \rbrace$
and let $\Omega^c$ be the connected component of $\Omega$ containing
$(x_-,x_+) \times \lbrace 0 \rbrace \times \bb{R}$. We define $\Omega^c$ to be the domain of $\sigma$, and we denote by $\Omega_{\e_{0}}^c$ the intersection of $\Omega^c$ with the hyperplane $\e=\e_{0}$.
\end{Def}
Notice that $\Omega_0^c$ is exactly the domain filled by the
characteristic lines (\ref{eq:chargen}) of the unperturbed equation with
initial point $x_0 \in (x_-,x_+)$.

\subsection{Proof of the validity of the variational string equation}

In Theorem \ref{thm:raymondGelfand}, we have shown that the variational derivative of a formal functional satisfies the linear PDE (\ref{eq:raymondGelfand}) with a forcing term proportional to the total time-derivative of the functional. Combining this result with \eqref{eq:stcomm}, we obtain that on the domain $\Omega^c$, the function $\sigma$ satisfies
\beq\label{eq:raymondgeneral}
\frac{\de\sigma}{\de t} =\sum_{i\geq 0} \mathcal{A}_i
\frac{\de^{i+1}\sigma}{\de x^{i+1}} + \e^{2N+2}\mathcal{R}, \qquad 
\sigma_{t=0}=0,
\eeq
where
$$\mathcal{A}_i(x,t,\e) = \left(\frac{\partial}{\partial u^{(i)}}\du{H_h}\right)_{|u=u(x,t;\e)} , \quad \mathcal{R}(x,t,\e) = r_{|u=u(x,t;\e)},$$
and  $r$ is as in (\ref{eq:stcomm}). Notice that
$r$ is a polynomial in $\e$, so that the forcing term in \eqref{eq:raymondgeneral} is $O(\e^{2N+2})$ provided $u$ is continuous as $\e \to 0$. If we consider just the explicit dependence on $\e^2$, then the linear differential operator $\sum_i \mathcal{A}_i\frac{\de^{i+1}}{\de x^{i+1}}$ acting on $\sigma$ is a \emph{polynomial} in $\e^2$. In particular, we can write it as:
\begin{eqnarray}\label{eq:Di}
\sum_{i=0}^m\mathcal{A}_i\frac{\de^{i+1}}{\de x^{i+1}} =
h''(u)\frac{\de}{\de x}+ \e^2 \Delta,
\end{eqnarray}
where $h''(u)$ is the same as in (\ref{eq:pde}) and $\Delta$ is a linear differential operator polynomial in $\e^{2}$, depending on $u$ and on its $x$-derivatives. Using (\ref{eq:raymondgeneral}) and above decomposition (\ref{eq:Di}), we prove that the string equation is
satisfied up to $O(\e^{2N+2})$.

\begin{thm}\label{thm:generalQ}
Suppose (\ref{eq:raymondgeneral}) depends on $u$ as well as on its
first $L$ derivatives, and let $[0,T]$ be a time interval such that the quantities
$$\frac{\partial^{j+2k} u(x,t,\e)}{\partial x^j \partial\e^{2k}}$$ 
exists and are continuous for $j\leq L, k\leq N$ for $(x,t,\e)\in \Omega^{c}$.  Then, we have that  $\sigma=O(\e^{2N+2})$ on
compact subsets of  $\Omega^c \cap \left\lbrace t \in [0,T] \right\rbrace$. Therefore, the conditions
$$\frac{\partial^{2k}\sigma(x,t,\e)}{\partial \e^{2k}}_{|\e=0}=0,$$
hold for $k=0,\dots,N$ and for any $t \in [0,T]$ .
\end{thm}
Before proving the theorem we need a lemma.

\begin{lem}\label{lem:defchar}
Under the same hypothesis as in Theorem \ref{thm:generalQ}, let $x(t;\e)$ be the solution of the Cauchy problem
\begin{equation}\label{eq:gencharlem}
\frac{d x}{d t}=-h''(u(x(t),t;\e)), \qquad x(t=0,\e)=x(0) \in (x_-,x_+),
\end{equation}
with $h''(u)$ as in (\ref{eq:pde}) and initial data $x(0)$ independent of
$\e$. Then, we have:
\begin{itemize}
\item[(i)] For any initial data $x(0)$, the point $(x(t;\e),t,\e)$
belongs to $\Omega^c_{\e}$ for any $t\in[0,T]$, provided $\e$ is small
enough.
\item[(ii)] Let $x_1,x_2$ be two solutions of (\ref{eq:gencharlem}) with
initial data $x_1(0) <x_2(0)$, then
$x_1(t;\e)<x_2(t;\e), \forall t \in [0,T]$ provided $\e$ is small
enough.
\item[(iii)] Let $K$ be a compact subset of $\Omega^c$ and $K_{\e'}:=K \cap
\left\lbrace t\in[0,T], \e^2 \leq \e'^2 \right\rbrace$. Then, there exist
$x_1(0),x_2(0)$, with $x_1(0)< x_2(0)$, and $\e'>0$ such that
$$K_{\e'} \subset \bigcup_{\e^2 \leq \e'^2 , \, t \in [0,T]} \left(
[x_1(t,\e),x_2(t,\e)],t ,\e \right) \subset \Omega^c \; . $$
Here $[x_1(t,\e),x_2(t,\e)]$ is the closed interval with extremes
$x_1(t,\e),x_2(t,\e)$.
\end{itemize}
\begin{proof}
We prove (i). Points (ii,iii) can be proven along the same guidelines. Along the curve $\dot{x}=-h''(u(x(t),t;\e))$, the function $u$ varies slowly,
namely
$$\frac{d u}{dt}=\e^2P(u,u_x,\dots,u^{(j)},\e) \, , \; \mbox{ for some } j
\leq L \; ,$$
where $P$ is a differential expression , smooth in $u$ and its first $j$
derivatives and polynomial in $\e^2$. Therefore
$u(x(t,\e),t,\e)-u(x(0),0,\e)=O(\e^2 t)$ uniformly in $t \in [0,T]$.
Hence if $\e$ is small enough, $u(x(t,\e),t,\e)$ belongs to $(a,b)$, the domain
of $f'$.
This proves the thesis.
\end{proof}
\end{lem}

\begin{proof}[Proof of Theorem \ref{thm:generalQ}]
The principle of the proof is to show that the $L^2$-norm  of $\sigma$ is
$O(\e^{2N+2})$ (uniformly) when evaluated on any subinterval
of $([x_1(t,\e),x_2(t,\e)]) $, for some function $x_1(t,\e)$, and $x_2(t,\e)$. Let us  compute how the $L^2$-norm varies on a time-dependent interval.
Using (\ref{eq:raymondgeneral}, \ref{eq:Di}), we get:
{\scriptsize \begin{eqnarray*}  \nonumber
& \norm{\frac{d}{dt}\int_{x_1(t)}^{x_2(t)}\sigma^2 dx} =
\norm{\int_{x_1}^{x_2}2 \sigma\,(h''(u)\sigma_x+\e^2 \Delta
\sigma+\e^{2N+2}
\mathcal{R}) dx+ \dot{x}_{2}\sigma^2(x_2)-
\dot{x}_{1}\sigma^2(x_1)} & \\ & \nonumber
=\big|-\int_{x_1}^{x_2}h'''(u)u_x\sigma^2+ 2 \sigma \big(
\e^2 \Delta \sigma+ \e^{2N+2} \mathcal{R}) \big)
dx  & \\& \nonumber
+ \left[\dot{x}_{2} +h''(u(x_2))\right]\sigma^2(x_2)- \left[
\dot{x}_{1}+h''(u(x_1))\right]\sigma^2(x_1) \big| &\\& \nonumber
\leq \sup_{x \in [x_1,x_2]}{\norm{h''(u)u_x}}
\int_{x_1}^{x_2}\sigma^2dx + \norm{\dot{x}_{2} +h''(u(x_2))}\sigma^2(x_2)+
\norm{\dot{x}_{1}+h''(u(x_1))}\sigma^2(x_1) + &\\& +
\big(\int_{x_1}^{x_2}\sigma^2 dx\big)^{1/2}  \left(
\e^2 \big( \int_{x_1}^{x_2}(\Delta \sigma)^2 dx\big)^{1/2}
+
\e^{2N+2} \big(\int_{x_1}^{x_2}\mathcal{R}^2 dx\big)^{1/2}   \right)
\; .& 
\end{eqnarray*}}
Here $\dot{x}_{i}=\frac{d x_i(t)}{dt}$ and the operator $\Delta$ is
defined in (\ref{eq:Di}). Choosing the boundary of the interval to evolve according to the law (\ref{eq:gencharlem}),
we get rid of the $\sigma^2(x_i)$ terms which are not controlled by the $L^2$ norm
\footnote{Note that in the $\e=0$, this is the law
of the characteristic lines \eqref{eq:chargen}.}:
\begin{eqnarray*}
&\norm{\frac{d}{dt}\int_{x_1(t)}^{x_2(t)}\sigma^2dx} \leq
\sup_{x \in [x_1,x_2]}{h''(u)\norm{u_x}}
\int_{x_1}^{x_2}\sigma^2dx  +&
\\ \nonumber &\big(\int_{x_1}^{x_2}\sigma^2 dx\big)^{1/2}
\left(
\e^{2} \big( \int_{x_1}^{x_2}(\Delta \sigma)^2 dx\big)^{1/2}
+
\e^{2N+2} \big(\int_{x_1}^{x_2}\mathcal{R}^2 dx\big)^{1/2}   \right) \; .&
\end{eqnarray*}
We define
$$\Pi = \bigcup_{\e^2 \leq \delta^2 , \, t \in [0,T]}([x_1(t,\e),x_2(t,\e)],t,\e)
$$
for some $\delta>0$. 
Applying the Gr\"onwall inequality (recalling that $\sigma_{|t=0}=0$ identically), we get the following bound for $\sigma$: 
\begin{eqnarray} \label{Qproofgen}
\norm{\int_{x_1(t)}^{x_2(t)}\sigma^2dx}^{\frac{1}{2}}  \leq \e^2  \,  A
( B+ \e^{2N} C) \frac{e^{DT}-1}{D} \, , \\ \nonumber
 A=\sup_{x,y \in \Pi}\norm{x-y}^{\frac{1}{2}} <\infty \; , B=
\sup_{\Pi}\norm{\Delta
\sigma(x,t)}  <\infty \, , \\ \nonumber
C=\sup_{\Pi}\norm{\cal{R}} <\infty \; ,  D=\sup_{\Pi}\norm{h''(u)u_x}
<\infty \; .
\end{eqnarray}
Since by Lemma \ref{lem:defchar} (ii) we have that $x_1(t) <x_2(t)$ for any $t$, provided $\e$ is small enough, the above estimate implies that $\sigma=O(\e^2)$ on $\Pi$. Therefore, we have $\sigma=O(\e^2)$ on any compact subset $K$ of $\Omega^c$. Indeed, the previous estimate holds for arbitrary $x_1(0),x_2(0)$ and
we know, by Lemma \ref{lem:defchar} (iii), that $K\cap\lbrace\e^2 \leq \e'^2 \rbrace \subset \Pi$ for some $x_1(0),x_2(0)$, and $\e'$. Therefore, by continuity, $\sigma=0$ on $\Omega^c_0$. We can now prove the thesis by induction. Suppose 
$$\frac{\partial^{2j} \sigma}{\partial \e^{2j}}_{|\e=0}=0,\qquad j=0,\dots,k$$
for some $k\leq N-1$. By the hypothesis on differentiability, 
$$\frac{\partial^{2j} \Delta \sigma}{\partial \e^{2j}}_{|\e=0}=0, \qquad j=0,\dots,k \,,$$
and so $\Delta \sigma=O(\e^{2k+2})$ on compacts of $\Omega^c$. Hence, (\ref{Qproofgen}) holds with $B= B' \e^{2k+2}$ for some constant $B'$. Therefore, 
$$\norm{\int_{x_1(t)}^{x_2(t)}\sigma^2dx}^{\frac{1}{2}} =O(\e^{2k+4})$$ 
for any $t \in[0,T]$. Reasoning as above, we obtain 
$$\frac{\partial^{2k+2} \sigma}{\partial\e^{2k+2}}=0,\qquad \text{on}\,\,\,\Omega^c_0.$$ 
The induction step is proven.
\end{proof}
\begin{rem}
In Theorem \ref{thm:generalQ} we can relax hypothesis (\ref{eq:stt0}).
Indeed the same Theorem holds if we choose
an initial data $\varphi$ such that
$\du{S_f}_{t=0,u=\varphi}=O(\e^{2N+2})$. The proof is valid with the
only immaterial modification of adding
an $O(\e^{2N+2})$ term on the right hand side of the estimate (\ref{Qproofgen}).
\end{rem}

\section{Particular examples and quasi-triviality for KdV}\label{sec:examples}
We have just shown that the String equation is approximately
satisfied, i.e. up to $O(\e^{2N+2})$ - provided the solution of the
Cauchy problem is regular
enough in the semiclassical limit. In fact, the hypothesis of
Theorem (\ref{thm:generalQ}) is rather reasonable.
We expect it to hold for very general Hamiltonian perturbations of Hopf
with a sufficiently regular initial data and for any time smaller than the critical time $t_c$. However, the only rigorous results in this direction are in our
previous paper \cite{mara11}, where we consider the Generalised KdV
case,
\begin{align}\label{eq:gkdv}
& u_t =h''(u)u_x+\sum_{j \geq 1}^{M} \e^{2j} \alpha_j u^{(2j+1)} ,
\qquad \alpha_j \in \bb{R}, \\ \nonumber
&u(t=0)=\varphi \in H^s(\bb{R}), \qquad s \geq 2M+1 \; .
\end{align}
Here $H^s$ is the Sobolev space of index $s \in \bb{R}$. We proved that the Cauchy problem has a well-posed semiclassical limit, namely
\begin{thm}\label{thm:mara}
Fix $K \in \bb{N}, K >1$ and let $u(x,t,\e)$ be the solution of the
Cauchy problem (\ref{eq:gkdv}).
If the initial data $\varphi$ belongs to $H^s(\bb{R})$ with $s\geq
3MK+2M+1$, then, for a time $T>0$ small enough, the
following partial derivatives exist and are continuous for every $x, t,\e$, with $
t\in [0,T]$:
\begin{equation*}
\frac{\partial^{i+j}u(x,t,\e)}{\partial x^{i}\partial \e^j}, \qquad
i< s-3MK-\frac{1}{2}, \quad  j\leq 2M K \; .
\end{equation*}
In particular, if $\varphi\in\bigcap_{s\geq 0}H^{s}(\mathbb{R})$ and for $(x,t)\in\bb{R}\times[0,T)$,  then the solution $u$ admits an asymptotic expansion as $\e\to 0,$ where all the coefficients are smooth functions of $x$ and $t$.
\end{thm}

This is a corollary of the more general Theorem 9 in \cite{mara11}. We can apply this Theorem to prove that the string equations is indeed
satisfied - up to the given error - by the solution the generalized KdV
(\ref{eq:gkdv}). From the theory developed in the previous section we can distinguish three cases:
\begin{itemize}
\item[A]Any generalised KdV equation admits a string equation of
order $O(\e^2)$.
\item[B]In some cases, for example if $h$ is cubic in $u$, then
the generalised KdV equations admits - at least - a variational string equation of order $\e^4$.
\item[C]Some particular equations, namely KdV or mKdV, which admit a variational string equation of arbitrary order.
\end{itemize}
In all cases, combining Theorems \ref{thm:generalQ} and \ref{thm:mara}
we can draw the conclusion that the string equation is indeed
satisfied up to the corresponding error, provided the initial data of the
equation belongs to a Sobolev space with high enough index. To give an example of how high the index must be, we consider the string equation of order $\e^4$ for KdV and mKdV. In this case, the linear equation (\ref{eq:raymondgeneral}) satisfied
by $\sigma$ depends on $u$ and its derivatives up to
$u^{(6)}$.
Due to Theorem \ref{thm:generalQ}, the String equation is satisfied up
to $O(\e^{6})$ if
$$\frac{\partial^{6}}{\partial x^{6}}\frac{\partial^{4}
u}{\partial \e^{4}}$$ exists and is continuous.
From Theorem \ref{thm:mara}, this holds provided $s>\frac{25}{2}$.

\begin{rem}
Suppose the Cauchy problem (\ref{eq:gkdv}) is globally well-posed if
$\e \neq 0$ and $\e$ small enough. Then
the time $T$ of Theorem \ref{thm:mara} can be chosen to be any time
smaller than the critical time of the solution of the Hopf equation. The above situation  happens in  the case of KdV and mKdV equations, and more generally,
if the Hamiltonian is positive definite or if $h(u)$ does not grow too
fast at infinity \cite{ka83}.
\end{rem}

As a corollary of our theory of the string equation we can prove that all the perturbative corrections of the Cauchy problem of KdV
are rational functions of the solution of the Hopf equation and its derivative.
In other words, there exists at any order in $\e$ a quasi-Miura transformation mapping the solution of the Hopf equation to the asymptotic expansion (for $\e$ small)
of the solution of KdV. This result was proved so far only for solutions of the Hopf equation with never vanishing first derivative and for initial data depending on $\e$ in a prescribed way, see \cite{du06}, \cite{lizh06} .  
\begin{thm}
Let $u(x,t,\e)$ be the solution of the KdV Cauchy problem
$$u_{t}=u\,u_{x}+\e^{2}u_{xxx},\qquad u_{t=0}=\varphi\in\bigcap_{s\geq 0}H^{s}(\mathbb{R}),$$
and let $t_{c}>0$ be the critical time of the Hopf equation with same initial data. Then,
for every $t\in [0,t_{c})$ and $x\in\bb{R}$, the solution $u$ possesses the asymptotic expansion
\beq\label{expufin}
u(x,t,\e)=\sum_{i\geq 0}\e^{2i}\,v^{i}(x,t),
\eeq
where 
\beq\label{ratvi}
v^{i}(x,t)=\xi^{i}(v^{0}(x,t),\de_{x}v^{0}(x,t),\dots,\de^{m_{i}}_{x}v^{0}(x,t),t),
\eeq
for some  $m_i.$ The functions $\xi^{i}(v^0,v^0_x, \dots, t), i \in \bb{N}$
are rational with respect to all the variables, and they do not depend on the initial data of the Cauchy problem.
\end{thm}
\begin{proof} From the results of \cite{mara11}, it follows that for the class of initial data in the hypotheses,
the solution $u$ admits an asymptotic expansion, for times smaller that $t_{c}$, of the form \eqref{expufin}. Here, $v^{0}$ is the solution of the Hopf Cauchy
problem with same initial data, and the smooth functions $v^{i}$ are the unique solutions (in $H^{\infty}$) of the semilinear equations
\begin{equation}\label{eq:transpoquasi}
 v^{i}_{t}=\sum_{j=0}^{i}\binom{i}{j}v^{j}v^{i-j}_{x}+i\,v^{i-1}_{xxx},\qquad
v^{i}|_{t=0}=0,\qquad i\geq 1 \; .
\end{equation}
Given any smooth initial data $\varphi$, we know that for any $N\geq 0$, the KdV equation admits a functional $S_f$ satisfying (\ref{2ste0}, \ref{2stt0}),
and condition (\ref{2stcomm}) with $K=2N+2$. By the results of Section \ref{sec:string}, the corresponding string equation is of the form
\begin{equation}\label{eq:kdvtriv}
x+u t= f'(u) + \sum_{j\geq1} \e^{2j} T_j(f''(u),\dots,f^{(k_j)}(u), u_x,\dots,u^{(2j)} ) \; ,
\end{equation}
for some functions $T_j$, depending rationally on $f'(u),\dots, f^{(k_j)}(u)$, and polynomially on $u_x,\dots, u^{(2j)}$. As usual,
here $f'$ is a local inverse of $\varphi$. Note that the $T_{i}$ are independent of the initial data.

We prove the theorem by induction. The zero-th correction is trivially true by taking $\xi^0(v^0)=v^0$.
Suppose now \eqref{ratvi} holds for all $\xi^j$ for $j=0,\dots,k$.

Let us restrict our attention to points $(x, t)$ with $t \in [0,t_c)$ and such that $v^{0}_{x}(x,t) \neq 0$. Indeed,
by Theorems \ref{thm:generalQ} and \ref{thm:mara} and for smooth initial data, in this case all the $\e$-derivatives of
(\ref{eq:kdvtriv}) vanish at $\e=0$. Taking the $(2k+2)$-th partial derivative of
(\ref{eq:kdvtriv}) with respect to $\e$, and evaluating at $\e=0$, we get
\begin{equation}\label{eq:derstringquasi}
(t-f^{(k+2)}(v^0))  \,v^{k+1} = R_{k+1}(f''(v^0), \dots, f^{(m)}(v^0), v^0_x, \dots,  v^k, v^k_x \dots) \; ,
\end{equation}
for some rational function $R_{k+1}$. Now since $x+v^0t=f'(v^0)$, then the quantities $f^{(i)}(v^{0})$ can be written as
rational functions of $t$ and of the derivatives of $v^0$. Due to the induction hypothesis, the same can be said for
$v^i, 1\leq i\leq k$ and their derivatives. Therefore, by (\ref{eq:derstringquasi}) $v^{k+1}$ is a rational function
of $t$ and of $v^0$ and its derivatives. We call this function $\xi^{k+1}$.

It remains to prove that $v^{k+1}(x,t) = \xi^{k+1}(v^0(x,t),v^0_x(x,t), \dots, t)$ also where $v^0_x=0$, provided
$(x,t) \in\bb{R}\times [0,t_c[$. This is easily accomplished by noticing that $\xi^{k+1}(v^0(x,t),v^0_x(x,t), \dots, t)$
satisfies (\ref{eq:transpoquasi}) -with $i=k+1$- for any $(x,t) \in\bb{R}\times [0,t_c[$.

In fact $\xi^{k+1}$ is analytic at $v^0_x=0$ for any value of the jets variables $v^0, v^0_{xx}, \dots$ and
provided $t \in [0, t_c)$, because due to Theorem \ref{thm:mara}
$v^{k+1}(x,t)$ is smooth for any initial data and for $(x,t) \in \bb{R}\times [0,t_{c})$.
Moreover, using the identity $v^0_t=v^0v^0_x$ and the induction hypothesis on $\xi^i, i=i,\dots,k$
the $k+1$-th transport equation (\ref{eq:transpoquasi}) reduces to an algebraic identity for $\xi^{k+1}$ which is then satisfied identically on the domain of analyticity
of $\xi^{k+1}$.
\end{proof}

\section*{Conclusions}

We have introduced a deformation of the method of characteristics valid for Hamiltonian perturbations
of a scalar conservation law in the small $\e$ (dispersionless, semiclassical) limit. Our method of analysis is based on the \lq variational string equation\rq, a functional-differential relation originally introduced by Dubrovin \cite{du06} in a particular case, of which we have laid the mathematical foundation. Starting from first principles, we have constructed the string equation explicitly up to the fourth order in perturbation theory and
we have shown that the solution of (the Cauchy problem of) the Hamiltonian PDE satisfies the appropriate string equation in the small $\e$ limit and for times smaller than the critical time. Moreover, we applied our construction to compute explicitly the first two perturbative corrections of the solution of the general Hamiltonian PDE, and we proved for KdV the existence of a global quasi-triviality transformation at any order in $\e$.\\

Yet many questions remain open about the string equation and its applicability, and more in general, about
the behaviour of solutions of Hamiltonian PDEs in the small $\e$ regimes. We plan to address them in future publications. A prominent open problem is the rigorous treatment of the dispersive shock phenomenon,
the behaviour of solutions of Hamiltonian PDEs in the small $\e$ regime and close to the critical time.
According to the Dubrovin universality conjecture, the solution of the generic
Hamiltonian PDE undergoes a phase transition from a regular behaviour to a highly oscillatory behaviour and this transition can be described in term of a special solution of the integrable ODE known as PI(2), which is the second member of the Painlev\'e I hierarchy. In fact, a formal multiscale expansion centered in a neighborhood of the critical time transforms any Hamiltonian PDE to KdV and the same multiscale expansion reduces the
fourth order string equation of KdV to PI(2). Remarkably, all higher order
terms of the string equation do not contribute to the expansion, so that the PI(2) correction is beyond all order in perturbation theory while based on a fourth order computation.\\

Our work provides another argument in favor of the universality conjecture.
Indeed, the same formal computation shows that the fourth order
string equation of any Hamiltonian PDE (which we constructed)
reduces to PI(2) in that multiscale expansion. To make such heuristic considerations into a  proof we plan to address the validity of the string equation in a neighborhood of the critical time.

Another problem we plan to address is more algebraic in nature: as we have shown, if we consider only functionals  with densities polynomial in the jets, then not all the Hamiltonian PDEs admit a string equation of order four. Indeed, a non-trivial obstruction arises at the fourth order and, presumably, higher order terms would impose more constraints on the Hamiltonian. We will investigate whether the obstruction(s) can be lifted by considering more general functionals, for example rational in the jet variables.

\paragraph{Acknowledgments} D.M. is since January 2012 a Scholar
of the Fundacao da Ciencia e Tecnologia, scholarship number $SFRH/ BPD/
75908/ 2011$. A.R. and D.M. thank, respectively, the Grupo de F\'isica Matem\'atica da Universidade de Lisboa and the SISSA Area of Mathematics for the kind hospitality. D.M. thanks Dr. Chaozhong Wu for having partially funded his visit to
SISSA through the SISSA grant \lq Giovani Ricercatori della SISSA\rq. We thank B. Dubrovin for pointing out that our results could
imply the Quasi-Triviality for KdV.

\bibliographystyle{plain}
\bibliography{biblio}

\def\cprime{$'$} \def\cprime{$'$} \def\cprime{$'$} \def\cprime{$'$}
  \def\cprime{$'$} \def\cprime{$'$} \def\cprime{$'$} \def\cprime{$'$}
  \def\cprime{$'$} \def\cprime{$'$} \def\cydot{\leavevmode\raise.4ex\hbox{.}}
  \def\cprime{$'$} \def\cprime{$'$} \def\cprime{$'$}
\begin{thebibliography}{10}

\bibitem{avno87}
V.~V. Avilov and S.~P. Novikov.
\newblock Evolution of the {W}hitham zone in {K}d{V} theory.
\newblock {\em Dokl. Akad. Nauk SSSR}, 294(2):325--329, 1987.

\bibitem{claeys2009}
T.~Claeys and T.~Grava.
\newblock Universality of the break-up profile for the {K}d{V} equation in the
  small dispersion limit using the {R}iemann-{H}ilbert approach.
\newblock {\em Comm. Math. Phys.}, 286(3):979--1009, 2009.

\bibitem{clgr11}
T.~Claeys and T.~Grava.
\newblock The {K}dv {H}ierarchy: {U}niversality and a {P}ainlev{\'e}
  {T}ranscendent.
\newblock {\em Int. Math. Res. Not.}, 220, 2011.

\bibitem{du06}
B.~Dubrovin.
\newblock On {H}amiltonian perturbations of hyperbolic systems of conservation
  laws. {II}. {U}niversality of critical behaviour.
\newblock {\em Comm. Math. Phys.}, 267(1):117--139, 2006.

\bibitem{dub08}
B.A. Dubrovin.
\newblock On universality of critical behaviour in {H}amiltonian {PDE}s.
\newblock In {\em Geometry, topology, and mathematical physics}, volume 224 of
  {\em Amer. Math. Soc. Transl. Ser. 2}, pages 59--109. ., 2008.

\bibitem{duzh01}
B.A. Dubrovin and Y.~Zhang.
\newblock {N}ormal forms of hierarchies of integrable {PDE}s, {F}robenius
  manifolds and {G}romov - {W}itten invariants.
\newblock arXiv:math/0108160v1, 2001.

\bibitem{gelfand76}
I.~M. Gel{\cprime}fand and L.~A. Diki{\u\i}.
\newblock A {L}ie algebra structure in the formal calculus of variations.
\newblock {\em Funkcional. Anal. i Prilo\v zen.}, 10(1):18--25, 1976.

\bibitem{ka83}
T.~Kato.
\newblock On the {C}auchy problem for the (generalized) {K}orteweg-de {V}ries
  equation.
\newblock In {\em Studies in applied mathematics}, volume~8 of {\em Adv. Math.
  Suppl. Stud.}, pages 93--128. Academic Press, New York, 1983.

\bibitem{lax75}
P.~D. Lax.
\newblock Periodic solutions of the {K}d{V} equation.
\newblock {\em Comm. Pure Appl. Math.}, 28:141--188, 1975.

\bibitem{liwuzh08}
Si-Qi Liu, C.-Z. Wu, and Y.~Zhang.
\newblock On properties of {H}amiltonian structures for a class of evolutionary
  {PDE}s.
\newblock {\em Lett. Math. Phys.}, 84(1):47--63, 2008.

\bibitem{lizh06}
Si-Qi Liu and Y.~Zhang.
\newblock On quasi-triviality and integrability of a class of scalar
  evolutionary {PDE}s.
\newblock {\em J. Geom. Phys.}, 57(1):101--119, 2006.

\bibitem{mara11}
D.~Masoero and A.~Raimondo.
\newblock {S}emiclassical limit for generalized {K}d{V} equations before the
  gradient catastrophe.
\newblock arXiv:1107.0461v2, 2011.

\bibitem{mo90}
G.~Moore.
\newblock Geometry of the string equations.
\newblock {\em Comm. Math. Phys.}, 133(2):261--304, 1990.

\bibitem{nov74}
S.~P. Novikov.
\newblock A periodic problem for the {K}orteweg-de {V}ries equation. {I}.
\newblock {\em Funkcional. Anal. i Prilo\v zen.}, 8(3):54--66, 1974.

\bibitem{no90}
S.~P. Novikov.
\newblock On the equation {$[L,A]=\epsilon\cdot 1$}.
\newblock {\em Progr. Theoret. Phys. Suppl.}, (102):287--292 (1991), 1990.

\bibitem{whi74}
G.~B. Whitham.
\newblock {\em Linear and nonlinear waves}.
\newblock Wiley-Interscience [John Wiley \& Sons], New York, 1974.
\newblock Pure and Applied Mathematics.

\end{thebibliography}

\end{document}